%% file: main.tex
\documentclass[10pt,conferences]{IEEEtran}
\pdfoutput=1
\usepackage[ruled,linesnumbered]{algorithm2e}
\usepackage{amssymb}
\usepackage[english]{babel}
\usepackage{caption}
\usepackage{subcaption}
\usepackage{float}
\usepackage{amsmath}
\usepackage{graphicx}
\usepackage{amscd}
\usepackage{amsthm}
\usepackage{verbatim}
\usepackage{tabularx}
\usepackage{proof}
\usepackage{alltt}
\usepackage{lipsum, color}
\usepackage{xspace}
\usepackage{fancybox}
\usepackage{wrapfig}
\usepackage{mathtools,cuted}
\usepackage{listings}
\usepackage{rotating}
\usepackage{xcolor,colortbl}
\usepackage[textsize=small,colorinlistoftodos]{todonotes}
\usepackage{algpseudocode}
\usepackage{lipsum, color}

\newtheorem{definition}{Definition}[section]

\newtheorem{theorem}{Theorem}[section]

\newcommand{\NNR}{\ensuremath{\mathbb{R}_{\geq 0}}}

\IEEEoverridecommandlockouts
\title{Combining Task-level and System-level Scheduling Modes for Mixed Criticality Systems}
\author{
\IEEEauthorblockN{Jalil Boudjadar$^1$}
\and
\IEEEauthorblockN{Saravanan Ramanathan$^2$}
\and
\IEEEauthorblockN{Arvind Easwaran$^2$}
\and
\IEEEauthorblockN{Ulrik Nyman$^3$}
\IEEEauthorblockA{~\\ $^1$ Aarhus University Denmark, $^2$ Nanyang Technological University, $^3$ Aalborg University Denmark}
}

\begin{document}
\maketitle

\begin{abstract}
Different scheduling algorithms for mixed criticality systems have been recently proposed. The common denominator of these algorithms is to discard low critical tasks whenever high critical tasks are in lack of computation resources. This is achieved upon a switch of the scheduling mode from \textbf{Normal} to \textbf{Critical}. We distinguish two main categories of the algorithms: \textit{system-level mode switch} and \textit{task-level mode switch}. System-level mode algorithms allow low criticality (LC) tasks to execute \textit{only} in normal mode. Task-level mode switch algorithms enable to switch the mode of an individual high criticality task (HC), from low (LO) to high (HI), to obtain priority over all LC tasks. This paper investigates an online scheduling algorithm for mixed-criticality systems that supports dynamic mode switches for both task level and system level. When a HC task job overruns its LC budget, then only that particular job is switched to HI mode. If the job cannot be accommodated, then the system switches to Critical mode. To accommodate for resource availability of the HC jobs, the LC tasks are degraded by stretching their periods until the Critical mode exhibiting job complete its execution. The stretching will be carried out until the resource availability is met. 
We have mechanized and implemented the proposed algorithm using Uppaal. To study the efficiency of our scheduling algorithm, we examine a case study and compare our results to the state of the art algorithms. 
\end{abstract}

\IEEEpubid{978-1-7281-2923-5/19/\$31.00 ˜\copyright˜2019 IEEE}
\IEEEpubidadjcol

\section{Introduction}
\input{introduction}

\section{Related Work}
\input{relatedwork}

\section{Multimode Scheduling of MCS}
\input{model}

\section{Schedulability Analysis}
\input{schedulability}

\section{Case Study}
\input{casestudy}

\section{Conclusion}
\input{conclusion}

\section*{Acknowledgment}{This work was in part funded by Independent Research Fund Denmark under grant number DFF-7017-00348, Compositional Verification of Real-time MULTI-CORE SAFETY Critical Systems.}
\bibliographystyle{abbrv}
\bibliography{Biblio}
\end{document}

%% file: introduction.tex
Modern embedded systems are achieved via the integration of different system components having different criticality levels on a single platform. Such systems are known by \textit{mixed criticality systems} (MCS). Examples are safety control systems in avionics \cite{Huyck12} and  automotive applications \cite{iso}. Mixed criticality systems are subjected to certifications dictated by the standards of different application areas, where different criticality levels require different assurance levels \cite{Ouedraogo18}. The consequences of missing a deadline vary in severity from task to task, according to the given criticality levels. It is therefore clear that highly critical components require a rigorous analysis to deliver a formal assurance about safety under error-free conditions, and the presence of certain defined errors maintains the behavior predictable \cite{Burns18}.  

During operation, it is important that critical tasks are supplied with sufficient computation resources to meet their time constraints. Running low critical tasks (\textbf{LC}) with the same privilege as high critical tasks (\textbf{HC}) enables the system functionality to be fully embraced \cite{Howard17,Su2015}, however this leads to potential violation of the critical tasks safety e.g deadline miss. An intuitive alternative is to prioritize critical tasks eternally over low/non critical ones by the use of \textit{criticality-as-priority}. Prioritizing critical tasks may require to discard low critical tasks. This may degrade the quality of service and functionality of the system \cite{Liu16,Jan13}.

Since Vestal's seminal work \cite{Vestal07}, different scheduling algorithms for mixed criticality systems have been introduced \cite{Lee17,Su13,Gettings15,Baruah12}. Such scheduling protocols rely on the assumption that a task can have different Worst Case Execution Time (WCET) bounds if one considers different confidence levels. This is due to the fact that determining the exact WCET of a task code is very pessimistic \cite{Burns10,Loefwenmark16}. A task's WCET can be bounded according to different confidence levels where the higher the confidence is the larger WCET will be \cite{Vestal07}.   

\IEEEpubidadjcol
Mixed criticality scheduling algorithms commonly use \textit{scheduling modes} to decide which tasks to consider for scheduling at any point in time \cite{Burns17}. In essence, a scheduling mode dictates the tasks that can be prioritized/ignored according to the actual workload, so that tasks of a given criticality level obtain privilege over the rest of the tasks regardless of the actual priorities. 
Within a given scheduling mode, tasks are scheduled according to the adopted scheduling policy. 

Scheduling algorithms for mixed criticality systems can be categorized, based on the type of mode switch scenario, in two groups: \textit{system-level mode} and \textit{task-level mode}. System-level mode scheduling algorithms \cite{Facs16,Lee17,zeroslack} employ two scheduling modes \textbf{Normal} and \textbf{Critical}. HC and LC tasks are equally scheduled under Normal mode. A mode switch from Normal to Critical happens whenever there is a potential insufficiency of computation resources due to one or more HC tasks exhibiting high confidence behavior, i.e., tasks run for more than their low confidence WCET. In Critical scheduling mode, LC tasks are either entirely dropped \cite{Baruah12,Ekberg14}, or run with a degraded service \cite{Facs16,Su13,Gettings15} to accommodate HC tasks. The system-level algorithms commonly penalize LC tasks \cite{Facs16,Lee17,zeroslack} as the system mode switch can be decided when a single HC task overruns its low confidence WCET.

Task-level mode switch \cite{Lee17,Huang13} is motivated by the fact that not necessarily all \textbf{HC} tasks exhibit high criticality behavior (largest WCET) at the same time. Thus,  only the HC tasks running high confidence WCET obtain priority over the rest of tasks. Each HC task runs in \textbf{LO} mode and switches to \textbf{HI} mode whenever it overruns its low confidence WCET. Such overruns can lead to insufficiency of computation resources where HC tasks running LO mode miss their deadlines if their priorities are lower than those of LC tasks. 


In this paper, we introduce a new elastic control-based scheduling algorithm by combining the aforementioned categories. The resulting algorithm relies on a \textit{job-level} mode switch technique, where the system mode switch occurs only when there is a \textbf{HC} task job, running \textbf{LO} mode, in risk to  miss its deadline due to a low priority. We restrict HC behavior to only the job that either exceeds its low confidence WCET or triggers a systems mode switch. On Critical mode, we run LC tasks under a degraded mode (periods stretching) rather than completely discarded. When the workload permits, LC tasks are compensated by shrinking subsequent periods to amortize the degradation. Our scheduling algorithm enables runtime resilience and recovery from overload transient scenarios. 

The rest of the paper is organized as follows: Section~\ref{sec:relatedwork} cites the relevant related work. Section~\ref{sec:combination} presents our multimode scheduling setting for MCS. In Section.~\ref{sec:analysis}, we show how to analyze the  schedulability. Section~\ref{sec:casestudy} is a case study. Finally, Section~\ref{sec:conclusion} concludes the paper.

%% file: relatedwork.tex
\label{sec:relatedwork}

Since Vestal's~\cite{Vestal07} seminal work on \textit{mixed-criticality} (MC) systems, several studies have been carried out in the recent past for MC scheduling. Most existing works on MC scheduling~\cite{Arvind13,Baruah12,Ekberg14,zeroslack,Park2011} rely on system-level mode switch i.e., when a HC task executes more than its low confidence WCET the remaining HC tasks  also simultaneously exhibit HC behavior. In order to guarantee resources for the HC tasks, many solutions employ a very pessimistic approach that completely discards all the LC tasks upon mode transition~\cite{Arvind13,Baruah12,Ekberg14}. There are some works to delay the dropping of LC tasks by postponing the mode switch instant~\cite{Santy12,Xiaozhe16,Hu16,Facs16}. Santy et al.~\cite{Santy13} and Bate et al.~\cite{Bate15} proposed some techniques to minimize the duration for which the system is in mode HI so that to reduce the non-service duration of LC tasks.

In this context, a plethora of studies has been carried out to improve the service offered to the LC tasks~\cite{Burns13,Xiaozhe16,Jan13,Su14,Su16,Su16_2,Baruah16,Liu16,Pathan17,Liu18,Fleming14,Gettings15,Huang15}. These approaches can be classified into four major categories:
\begin{enumerate}
\item \emph{Elastic Scheduling}. The dispatch frequency of LC tasks is reduced (extending their periods) in the HI mode~\cite{Burns13,Jan13,Su14,Su16,Su16_2,Facs16}.
\item \emph{Imprecise Computation/Reduced Execution}. LC tasks are executed with reduced execution budget when the system is in mode HI~\cite{Burns13,Baruah16,Liu16,Xiaozhe16,Pathan17,Liu18}.
\item \emph{Selective Degradation}. Depending on the budget availability in the HI mode, only a certain subset of LC jobs/tasks are executed~\cite{Xiaozhe16,Fleming14,Gettings15}.
\item \emph{Processor speedup}. Huang et al.~\cite{Huang14,Huang15,Boudjadar17} proposed a dynamic processor speedup technique to guarantee resources for HC tasks instead of degrading the service to the LC tasks in the HI mode.
\end{enumerate}

However, all the above works employ an impractical assumption that all the HC tasks in the system simultaneously exhibit HC behavior. On the contrary, there are very few works that relax the system-level mode switch assumption and employ task-level mode switch~\cite{Huang13,Ren15,Xiaozhe15,Lee17}. Task-level mode switch algorithms restrict the impact of HC tasks exceeding their low confidence WCET and limit the service degradation of LC tasks. 

Huang et al.~\cite{Huang13} proposed a constraint graph to map the execution dependencies between HC tasks and LC tasks: when a HC task exhibits HC behavior only the LC tasks connected to it are dropped. However, in their analysis they consider all HC tasks utilize their high confidence WCET. Ren et al~\cite{Ren15} proposed a similar technique in which each HC task is grouped with some LC tasks and only these tasks are affected if that particular HC task exhibits HC behavior.

Gu et al~\cite{Xiaozhe15} presented a hierarchical component-based scheduling technique that allows multiple HC tasks to be grouped within a component. If any HC task in a component switches to HI mode, all the HC tasks in the component are run with their high confidence WCET and the LC tasks within that component are discarded. The authors also limit the number of components that can safely switch to HI mode using a tolerance parameter to trigger the system  mode switch.

Erickson et al.~\cite{Erickson15} proposed a scheduling framework for multicore mixed criticality systems to recover from transient overload scenarios. The recovery relies on scaling the task inter-release times to reduce the jobs frequency. The underlying schedulability analysis requires that all tasks must run the WCETs of the same confidence level, which implies to rerun the analysis for each criticality level separately. Compared to that, our schedulability analysis is performed across different criticality levels at once. 

Lee et al.~\cite{Lee17} proposed an online schedulability test for task-level mode switch and an adaptive runtime task dropping strategy that minimizes LC task dropping. However, they consider all the jobs of a HC task exhibit HI mode behavior which may be a pessimistic assumption. Recently, Papadopoulos et al.~\cite{Papadopoulos18} presented a control approach to achieve resilience in MC systems. HC tasks and LC tasks are executed using a server-based approach and based on the runtime property of the tasks the budget allocated to these servers is dynamically varied. When a HC server exhibits HC behavior, the LC servers are under-scheduled to meet the demand of HC servers. We rely on the same control-based mechanism to achieve LC task periods stretching, however we  compensate such a degradation by shrinking LC task periods whenever the HC tasks workload permits.  

In contrast to the above studies, we propose a dynamic mode switching algorithm that allows both task-level and system-level mode transitions. In particular, we restrict the HC behavior to only the job that either exceeds its low confidence WCET or triggers a systems mode switch. At the same time, we offer a minimum service to all LC tasks in the Critical mode using elastic scheduling instead of dropping them.

%% file: model.tex
\label{sec:combination}

In this section, we combine system-level and task-level scheduling modes to produce a multimode scheduling algorithm for MCS. Our mixed criticality scheduling algorithm enables efficient mode switches for HC tasks, by predicting the workload causing HC tasks to fail. 

\subsection{System model}
We consider deadline-implicit periodic task systems with two distinct criticality levels: high (HC) and low (LC), so that each mixed criticality (MC) task can be a LC or HC. By \textit{default criticality}, we refer to the criticality level assigned to a given task at the design stage (constant). The \textit{runtime criticality} of a task is in fact the (dynamic) criticality level assigned to the task according to the scheduling mode and/or task behavior. 

\paragraph{Assumptions} We consider the following assumptions:
\begin{itemize}
\item Tasks are preemptible.
\item All tasks are assigned a static criticality level (LC or HC) by design, called default criticality.
\item The execution of a HC task must not be discarded under any runtime circumstances.
\item The runtime criticality of a LC task can never be upgraded to HC.
\item LC tasks stick always to their low confidence WCET.
\item There is no dependency between LC and HC tasks.
\end{itemize}
 
\paragraph{Notations}
\begin{itemize}
\item We use $\pi_i$ to refer to a single task, and $\Pi$ to refer to the set of tasks. 
\item $Mode(t)\in\{Normal, Critical\}$ states the system scheduling mode at time point $t$. 
\item To track the mode of individual HC tasks over runtime, we introduce a function $\Omega: \{\pi_i \mid \chi_i=\textbf{HC}\} \times \NNR \rightarrow\{\textbf{HI}, \textbf{LO}\}$. For the sake of notation, we write $\Omega(\pi_i,t)$ for the mode of task $\pi_i$ at time point $t$. 
\end{itemize}

\begin{definition}[Tasks]
A task $\pi_i$ is given by $\langle T_i,C_i^l,C_i^h,\chi_i, \rho \rangle$ where:
\begin{itemize}
	\item $T_i$ is the task period.
	\item $C_i^l \in \NNR$ and $C_i^h\in \NNR$ are the worst case execution time for low and high confidence levels respectively. We assume that $C_i^h\ge C_i^l$ for HC tasks, and $C_i^h= C_i^l$ for LC tasks.
  \item $\chi_i \in \{\textbf{LC},\textbf{HC}\}$ is the default (constant) criticality of the task.
	\item $\rho$ is the task priority.
\end{itemize}
The task runtime mode $\Omega()$ will be updated on the fly according to the actual task execution budget.
\end{definition}

We distinguish between the task mode $\Omega(\pi_i,t)$, which is individual for each task, and the system scheduling mode $Mode(t)$. A task scheduling mode is driven by its execution time, so that whenever the execution violates the low confidence WCET $C_i^l$ the task mode is elevated to \textbf{HI}. The individual mode of a HC task switches independently. The overrun of $C_i^l$, by a HC task, is considered to be non-deterministic. 

The system scheduling mode is common for all tasks. It determines the tasks that are allowed to execute, and the main scheduling criterion (criticality, priority or both). Under \textbf{Normal} mode, all ready tasks are equally scheduled according to the adopted scheduling policy. However, when the system mode is \textbf{Critical} criticality levels are used as the main scheduling criterion to arbitrate tasks. If two tasks have the same criticality level, then we refer to their actual priorities. In such a scheduling mode, \textbf{LC} tasks may not be scheduled given their low criticality level. A stretching of the LC task periods is applied while the system runs in mode Critical. Thus, reducing the utilization of LC tasks to accommodate HC tasks. Whenever the system scheduling mode returns to Normal, the periods of LC tasks are then shrunk to amortize the delays created by the stretching. The shrinking can start only after LC tasks complete the jobs of the periods experienced a stretching.

Taskset $\Pi$ will be scheduled by the real-time operating system according to a scheduling function $Sched$. In fact, $Sched()$ implements an actual static priority-based scheduling policy such as Fixed Priority scheduling (FP). 
 \[\mathit{Sched}: 2^{\Pi} \times \NNR \rightarrow \Pi\]

In a similar way, we define a (\emph{Intermediate}) scheduling function $Sched_I(\Pi,t)$ which employs both task mode and priority. Thus, a task gets scheduled at a given time point $t$ if it has either a higher task mode\footnote{We consider that $HI>LO$, but HC tasks running in mode LO are comparable to LC tasks.} compared to any ready task, or the same task mode but a higher priority.   
\vspace{-2mm}

\[
\noindent
\begin{array}{ll}
\hspace{-2mm} \mathit{Sched}_I(\Pi,t)= \pi_i \mid & \hspace{-3mm}
 \mathit{Ready}(\pi_i,t) \wedge \forall \pi_j \in \Pi ~ 
  \mathit{Ready}(\pi_j,t) \Rightarrow \\ & \hspace{-12mm} 
\left \{\begin{array}{l}  
\Omega(\pi_j,t)< \Omega(\pi_i,t) \\
\vee \\
\Omega(\pi_j,t)=\Omega(\pi_i,t) \wedge \mathit{Sched}(\{\pi_i,\pi_j\},t)=\pi_i
\end{array}
\right. 
\end{array}
\]  
\noindent where $\mathit{Ready}(\pi_i,t)$ is a predicate stating whether a given task is ready at a given time point. As a third stage, we define a more restrictive scheduling function $Sched_C()$ which employs \emph{Criticality} level, task mode and priority to decide which task to be scheduled at any point in time.
\vspace{-3mm}

\[
\begin{array}{ll}
\hspace{-2mm} \mathit{Sched}_C(\Pi,t)= \pi_i \mid & \hspace{-3mm} \mathit{Ready}(\pi_i,t) \wedge \forall \pi_j \in \Pi ~ 
  \mathit{Ready}(\pi_j,t) \Rightarrow \\ & \hspace{-6mm} 
\left \{\begin{array}{l}  
\chi_j<\chi_i \\
\vee \\ 
(\chi_j=\chi_i) \wedge \Omega(\pi_j,t)< \Omega(\pi_i,t) \\
\vee \\
(\chi_j=\chi_i) \wedge (\Omega(\pi_j,t)=\Omega(\pi_i,t)) \\ \hspace{13mm} \wedge ~\mathit{Sched}(\{\pi_i,\pi_j\},t)=\pi_i
\end{array}
\right. 
\end{array}
\] 

The utilization of $\mathit{Sched}_I()$, $\mathit{Sched}_C()$ and $\mathit{Sched}()$ is described in the next sections. In the rest of this section, we present our task-level and system-level mode switches and how to combine both modes to achieve a more flexible scheduling.

\subsection{Task-level mode switch}
\input{TaskLevelMode}

\subsection{System-level mode switch}
\input{SystemLevelMode}

\subsection{Multimode Scheduling Algorithm}
\input{algorithm}

%% file: TaskLevelMode.tex

\paragraph{Low criticality tasks behavior}
Low criticality tasks are not concerned by the task mode switch because they are not concerned by rigorous certification as high criticality tasks. They are also assumed to run always the same WCET, i.e. $C^l=C^h$. Figure~\ref{fig:LCbehavior} illustrates the LC tasks behavior. In fact, LC tasks execute regularly next to HC tasks as long as the system scheduling mode is Normal. Under that context LC tasks are equally scheduled, using $\mathit{Sched}()$, as HC tasks running in mode LO. 

\begin{figure}
\centering
\caption{Low criticality task behavior}
\label{fig:LCbehavior}
\includegraphics[scale=0.43]{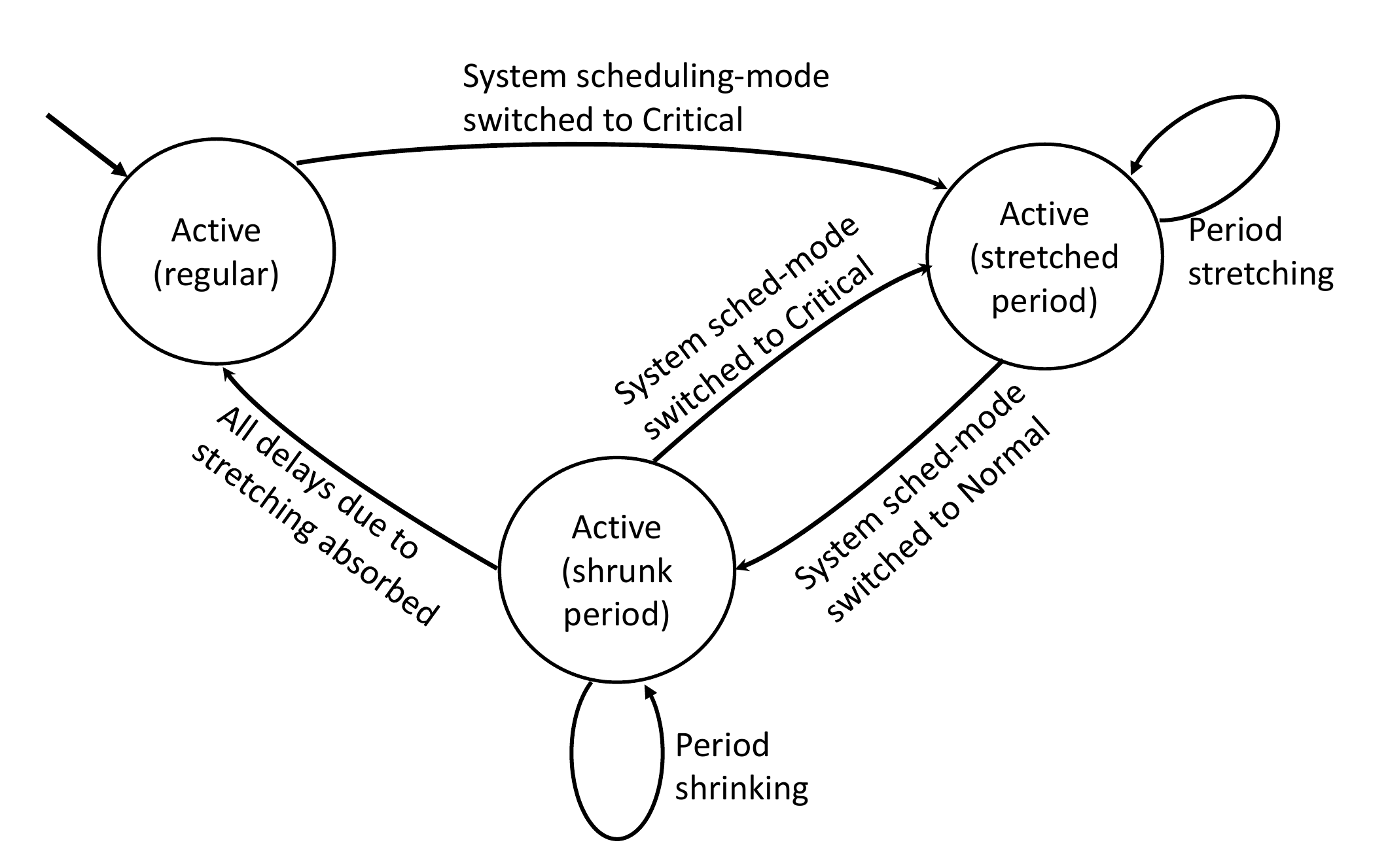}
\vspace{-9mm}
\end{figure}

Upon a switch of the system mode to Critical, the current job periods of LC tasks are stretched to reduce their utilization and the frequency of releasing new jobs. The system is then declared to be performing a stretching pattern. We introduce a variable ${\mathcal P}\in \{Stretching, Shrinking, Regular\}$ to store the current system pattern. 

 To track the stretching duration, we use a variable $s$ which indicates how much an LC task needs to be compensated in order to absorb the delays caused by the stretching. The stretching of LC tasks is a degraded operation mode. 

Whenever the system scheduling mode is back to Normal and the current stretched periods expire, the stretching is interrupted and the LC tasks can then execute regularly. To amortize the slack time created by stretching, the scheduler applies a shrinking to LC task periods \footnote{The system pattern is then updated accordingly, ${\mathcal P}=Shrinking$.}. The shrinking pace depends on the system workload and the LC task periods length. The fewer HC tasks run $C^h$ the larger the shrinking will be. 
Once all the delays introduced due to stretching are amortized, LC tasks run regular periods\footnote{${\mathcal P}=Regular$.}. 
 
\begin{figure}
\begin{center}
\caption{Stretching/shrinking of LC task periods}
\label{fig:LCruntime}
\includegraphics[scale=0.56]{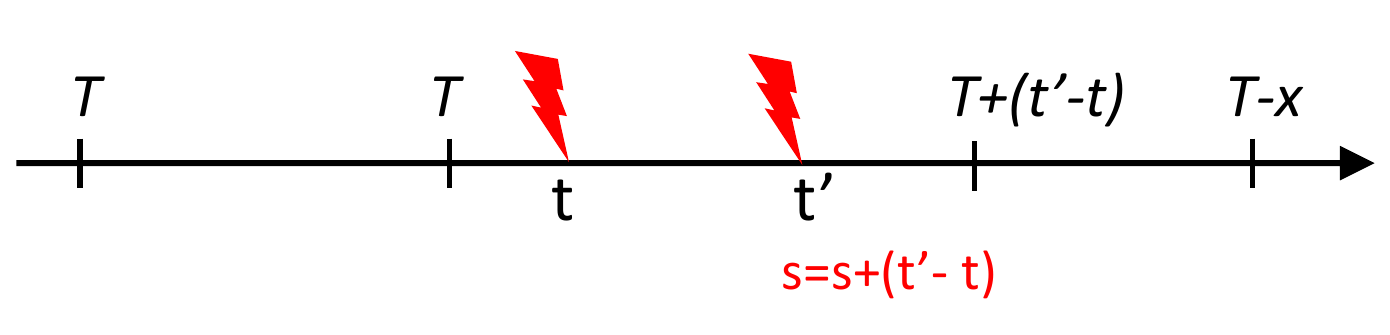}
\vspace{-9mm}
\end{center}
\end{figure}
       
Figure~\ref{fig:LCruntime} depicts an example of stretching and shrinking operations for an LC task period. Within the initial period, the task executes normally. After releasing the second period, a system mode switch (from Normal to Critical) happens at time $t$ causing the period to be stretched until time instant $t'$ where another system mode switch (Critical to Normal) occurs. The stretching duration $t'-t$ is accumulated in $s$. The third period will then be shrunk with $0\le x \le s$ to absorb the delay $s$. If the delay $s$ is not completely absorbed in one period, subsequent periods will be shortened accordingly. Formal calculation of the stretching/shrinking durations is provided in Section~\ref{sec:systemmodes}

Given that $C^l$ and $C^h$ are equal for each LC task, we simply write $C$. The utilization of a LC task is defined as follows:
\begin{itemize}
	\item Regular activation: $U_{L_{i}} = \frac{C_i}{T_i}$
	\item During shrinking with a duration $\delta$: $U_{L_{i}}^{\delta} = \frac{C_i}{T_i-\delta}$ such that $C_i\le (T_i-\delta)$.
\end{itemize}


\paragraph{High criticality tasks behavior}
Each individual HC task starts at mode LO and can change its mode independently from the rest of tasks. By default, on the release of a new period the HC task runs LO mode and whenever $C^l$ overrun happens the task mode switches to HI \cite{Lee17}. Such a task mode is maintained until the expiry of the given period. The budget overrun is \textit{non-deterministic}. Figure~\ref{fig:HCbehavior} illustrates the mode switches of HC tasks.


Whenever a HC task switches to mode HI, $\Omega(\pi_i,t)=\textbf{HI}$, it obtains the scheduling privilege over all LC tasks. Besides, a HC task running in HI mode has priority over all HC tasks running in LO mode. Among the HC tasks running HI mode, the task having the highest priority is scheduled first. Function $Sched_{I}()$ is used to schedule tasks according to these criteria.  
   
\begin{figure}
\begin{center}
\caption{High criticality task behavior}
\label{fig:HCbehavior}
\includegraphics[scale=0.43]{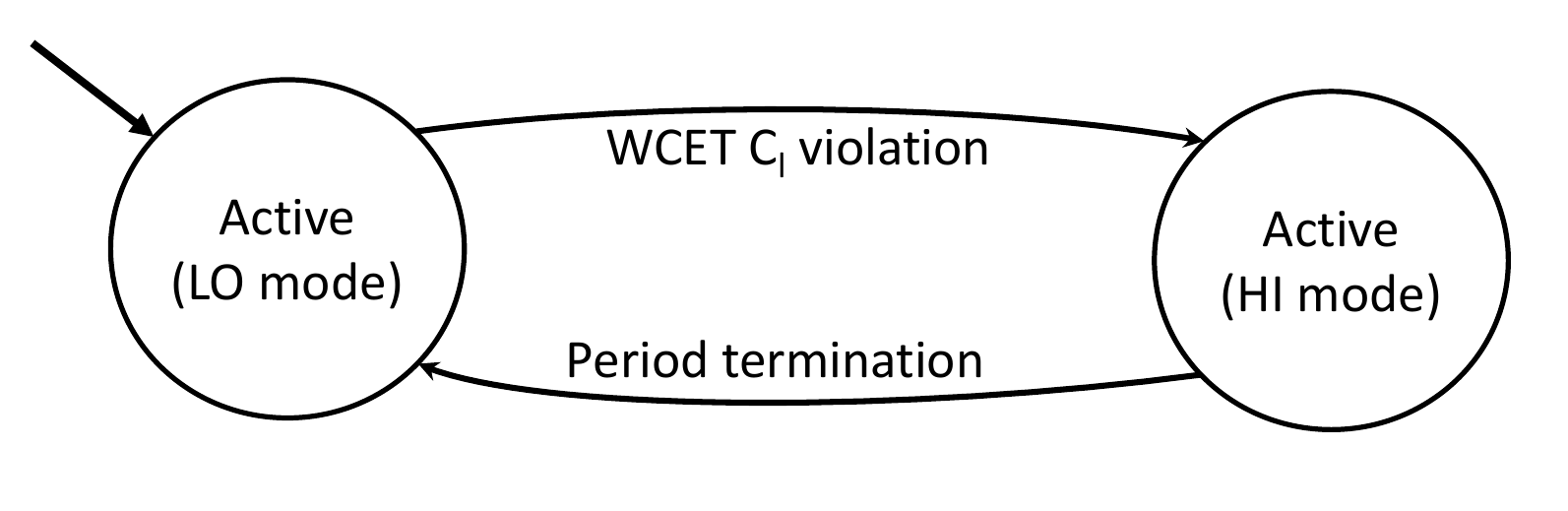}
\end{center}
\vspace{-11mm}
\end{figure}
 
However, given that HC tasks running LO mode do not have privilege over LC tasks, a HC task can miss its deadline under LO mode in case there is a lack of computation resources to execute both HC and LC tasks. This can be considered to be the major drawback of both task-level and system-level scheduling algorithms of mixed criticality systems. To circumvent this issue, our scheduling algorithm can assign a HC task running in LO mode the privilege over LC tasks even though it does not overrun its low confidence WCET $C^l$. 

We define the utilization of a HC task $\pi_i$ running mode HI, respectively mode LO, by: 
\[U_{H_i}=\frac{C^H_i}{T_i},~ \text{respectively}~~ U_{L_i}=\frac{C^L_i}{T_i}\]

 We also use $U_{L}$ to refer to the utilization of LC tasks. To specify the task mode switches, we introduce the following functions:
\begin{itemize}
\item $Status(\pi_i,t)\in\{{Ready, Running, Done}\}$ returns the status of any task $\pi_i$ at any point in time $t$.
\item $\Lambda(\pi_i,t)$ returns the budget consumed at time $t$ by the current release of a task $\pi_i$. $\Lambda(\pi_i,t)$ is not accumulative, i.e., it resets to zero upon each period release.
\end{itemize}

Formally, the runtime mode of a high criticality task switches from LO to HI as follows:

\begin{center}
\fbox{
\parbox[t][3.4em][t]{0.47\textwidth}{
$\frac{\begin{array}{c} \forall ~ \pi_i\in\Pi \mid \chi_i=HC,~\forall t \mid \\
        Status(\pi_i,t)\neq Done \wedge \Lambda(\pi_i,t) \ge C^l_i \wedge \Omega(\pi_i,t)=LO
				\end{array}}
				{\begin{array}{c} \Omega(\pi_i,t)\mapsto HI\end{array}}$
}}
\end{center}

Accordingly, the runtime criticality of a HC task returns to LO mode whenever its period expires as shown below.

\begin{center} 
\fbox{
\parbox[t][3.3em][t]{0.47\textwidth}{
$\frac{\begin{array}{c} \forall ~ \pi_i\in\Pi \mid \chi_i=HC,~\forall t \mid \Omega(\pi_i,t)=HI ~~\wedge \\ 
Status(\pi_i,t)= Done \wedge t ~\%~ T_i=0 
				\end{array}}
				{\begin{array}{c} \Omega(\pi_i,t)\mapsto LO \end{array}}$
}}
\end{center}
\vspace{2mm}

$\%$ is the arithmetic modulo operator. One can see that the task-level mode switch relies on the violation of $C^l$ and does not guarantee the feasibility of HC tasks running LO mode. 

%% file: SystemLevelMode.tex
\label{sec:systemmodes}
As stated earlier, the task level mode can be used to prioritize HC tasks running in HI mode. 
The drawback of the task level scheduling mode is then \textit{how to prioritize a HC task running a LO mode when the system workload lacks computation resources}. To circumvent this drawback, our system level mode complements the task level mode and  constrains the classic system level mode switches with the workload of HC tasks running both \textbf{LO} and \textbf{HI} modes equally. Let us illustrate the aforementioned drawback scenario for the system of Table~\ref{tab-example}. 

\begin{table}[htb!]
\vspace{-2mm}
\begin{center}
\caption{Example of a failure case for both system and task level scheduling modes}
\label{tab-example}
\begin{tabular}{| c | c | c | c | c |c|}
\hline
Task &    T & $C^l$ & $C^h$ & $\chi$ & $\rho$ \\ \hline
$\pi_1$ & 20 & 5 & 7 & HC & 2 \\ \hline
$\pi_2$ & 20 & 5 & 6 & HC & 4 \\ \hline
$\pi_3$ & 20 & 5 & - & LC & 1 \\ \hline
$\pi_4$ & 20 & 4 & - & LC & 3 \\ \hline
\end{tabular}
\end{center}
\vspace{-3mm}
\end{table}

\begin{figure}
\begin{center}
\caption{Runtime example for the system in Table.~\ref{tab-example}}
\label{fig:drawback}
\includegraphics[scale=0.42]{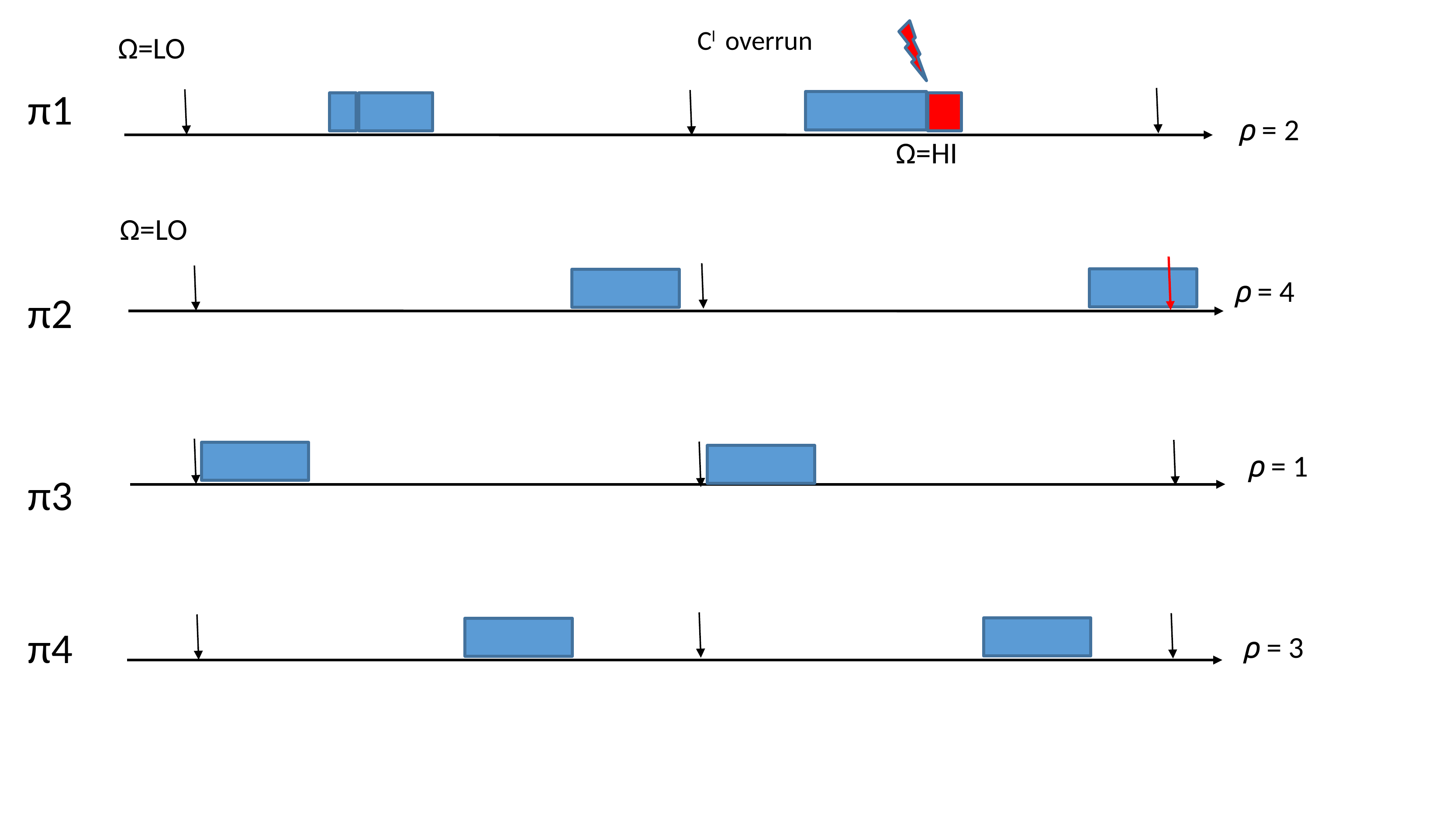}
\vspace{-5mm}
\end{center}
\end{figure}

Figure~\ref{fig:drawback} depicts a runtime example. On the first period, tasks execute according to the order of their priorities. On the second period, $\pi_1$ violates its $C^l=5$ and runs for two extra time units. This delays $\pi_4$, which in turn delays $\pi_2$ due to its lower priority. In the end, $\pi_2$ misses its deadline with one time unit. This scenario could be avoided if one would account for the feasibility of $\pi_2$, at the time point when $\pi_1$ violates $C^l$, and elevate its priority immediately. Thus, $\pi_2$ would execute before $\pi_4$ and meets its deadline.  

To summarize, our system level scheduling mode monitors the workload, for both LC and HC tasks, online and decides when to prioritize HC tasks over all LC tasks regardless of the HI/LO task modes. The system scheduling mode is effectively switched from Normal to Critical if the actual workload of LC tasks and HC tasks exceeds the resource supply for a time interval starting at the actual time point. 

\begin{figure}
\begin{center}
\caption{System scheduling mode behavior}
\label{fig:Systembehavior}
\includegraphics[scale=0.46]{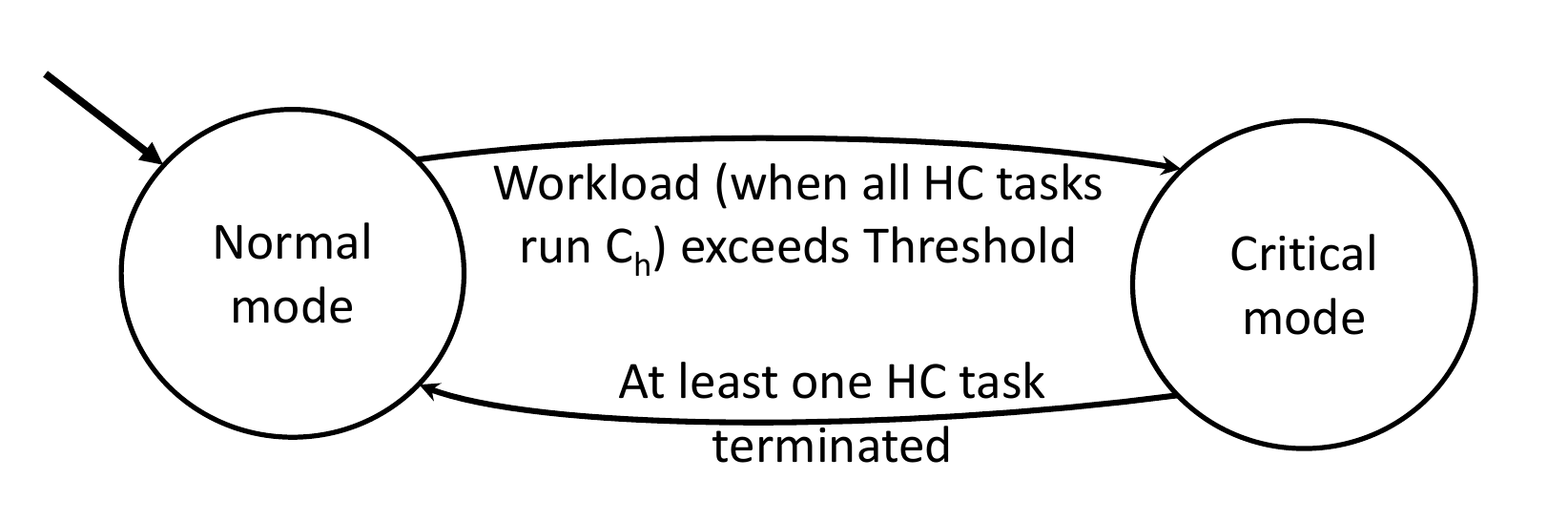}
\vspace{-7mm}
\end{center}
\end{figure}

Figure~\ref{fig:Systembehavior} shows the system mode behavior. The system is initially at Normal mode, and transits to Critical mode when the resource demand exceeds the resource supply. LC task periods are stretched accordingly, thus reducing their utilization, to make room in the schedule for HC tasks at least for their low confidence WCET $C^l$. Whenever the workload of HC tasks is relaxed, the system switches back to Normal and LC tasks can then be compensated to absorb the delay caused by stretching. 

We define the workload function $\Psi(\pi_i, [a,b])$ of a task $\pi_i$ over a time interval $[a,b]$ to be the amount of resource that can be requested by $\pi_i$. Such a workload includes the remaining execution time at time point $a$ for the current job plus the jobs to be potentially released until time instant $b$. We distinguish between $\Psi^H()$ and $\Psi^L()$ according to the task criticality and modes.
 \vspace{-6mm}

\[ 
\Psi^H(\pi_i, [a,b])= \hspace{-1mm} \left \{\hspace{-2mm} \begin{array}{l}
\small{C^h_i-\Lambda(\pi_i,a) + U_{H_i}\cdot T_i \cdot \lceil\frac{b-a}{T_i}\rceil~~\text{If}} 
\\~~~~~~~~~~~~~~~~~~~~~~~~~~~~~~ \small{(b-a)\%T_i\ge C_i^h } \\
\small{C^h_i-\Lambda(\pi_i,a) + U_{H_i}\cdot T_i \cdot \lfloor\frac{b-a}{T_i}\rfloor} 
~~ \footnotesize{\text{Otherwise}} 
\end{array}
\right .
\]

\[
\Psi^L(\pi_i, [a,b])= \hspace{-1mm} \left \{\hspace{-2mm} \begin{array}{l}
\small{C^l_i-\Lambda(\pi_i,a) + U_{L_i}\cdot T_i \cdot \lceil\frac{b-a}{T_i}\rceil~~\text{If}}  \\ ~~~~~~~~~~~~~~~~~~~~~~~~~~~~~~
\small{~(b-a) \%T_i\ge C_i^l } \\
\small{C^l_i-\Lambda(\pi_i,a) + U_{L_i}\cdot T_i \cdot \lfloor\frac{b-a}{T_i}\rfloor} ~~\footnotesize{\text{Otherwise}} 
\end{array}
\right .
\]

We define the workload of HC tasks having a high criticality than $\pi_i$ for the time interval $[t,T_i]$ as follows:
\[\vspace{-2mm}
W_{H}^h(\pi_i,t)=\sum\limits_{\pi_j\mid \chi_i=HC \wedge \Omega(\pi_j,t)=HI} \Psi^H(\pi_j, [t,T_i])
 \]
Implicitly, the time interval $[t,T_i]$ is the duration left to the expiry of the last period released by task $\pi_i$ before time point $t$, i.e. $[t~\%~T_i, T_i]$. Thus, we avoid writing the conversion absolute-relative time.    
In a similar way, we calculate the workload of HC tasks running LO mode and having higher priority than $\pi_i$, for time interval $[t,T_i]$ as follows:
\[ W_{L}(\pi_i,t)=\sum\limits_{\pi_j\mid \chi_j=LC \wedge \pi_j\in hp(\pi_i,t)} \Psi^L(\pi_j, [t,T_i])
\]

\noindent where $hp(\pi_i,t)$ is the set of tasks having a higher priority than $\pi_i$ at time point $t$. Finally, the workload of LC tasks having a higher priority than $\pi_i$ is given by: 
\vspace{-2mm}

\[W_{H}^l(\pi_i,t)=\sum\limits_{\pi_j\mid \chi_j=HC \wedge \pi_j\in hp(\pi_i,t) \wedge \Omega(\pi_j,t)=LO} \Psi^L(\pi_j, [t,T_i])
\]     

We define \texttt{DEM}$(\pi_i,t)$, an upper bound on the resource demand over a given time interval \cite{Baruah90}, of a HC task running in LO mode at any time point $t$ till the expiry of that period to be the remaining budget of such a task for the given period plus the workload of tasks having either a higher criticality or a higher priority. Namely, these are LC tasks having a higher priority, HC tasks running HI mode and HC tasks running LO mode but having higher priority than task $\pi_i$. 
\vspace{-1mm}
\[
\texttt{DEM}(\pi_i,t)=\hspace{-1mm} W_{H}^h(\pi_i,t) + W_{H}^l(\pi_i,t) + W_{L}(\pi_i,t) +C^L_i-\Lambda(\pi_i,t)
\]
 
One can see that we distinguish between HC tasks running HI, and HC tasks running LO and having higher priority than a given task. This is in fact to avoid counting the tasks satisfying both conditions twice in the workload. Given that the maximum resource amount that can be supplied to the task set during a time interval $[a,b]$ is $b-a$, the system scheduling mode switches from Normal to Critical if the workload exceeds (or is going to exceed) the resource supply.

\vspace{2mm}
\begin{center}
\fbox{ 
\parbox[t][3.3em][t]{0.4\textwidth}{
$\frac{ \begin{array}{c} \exists t~\pi_i\mid \chi_i=HC ~~\wedge \\ \Omega(\pi_i,t)=LO ~ \wedge~~  \texttt{DEM}(\pi_i,t) \ge  T_i-(t~\%~T_i) \end{array}}
{\begin{array}{c} Mode(t)\mapsto \textbf{Critical} \end{array}}$
}}
\end{center}

One can see that the load calculation, as a ground for the system mode switch, is performed on the time interval of the actual trigger task rather than classic entire busy period. This is in fact to reduce the over-approximation of the workload, given that low confidence WCET violation is non-deterministic, and deliver an exact load calculation. 

Once the system scheduling mode is switched to Critical, the periods of LC jobs will be extended with the time left of the current release ($T_i-(t~\%~T_i)$) of the HC task ($\pi_i$) causing the mode switch. 

Let us call the HC task causing the actual system mode switch a \textit{trigger} ${\mathcal T}$, and ${\mathcal S}$ the relative time instant of the corresponding mode switch \footnote{For the sake of notation, we consider ${\mathcal S}$ to be a time instant relative to the current release of the trigger task so that we avoid the conversion relative-absolute time.}. Thus, we simply write ${\mathcal T}(\pi_i,{\mathcal S})$ for a task $\pi_i$ being a trigger at time ${\mathcal S}$. In Critical mode, the system uses $Sched_C()$ to schedule tasks rather than $Sched()$ so that LC tasks do not have a chance to execute before any HC task regardless of the HC task mode and priority. This does not mean that LC tasks are discarded but rather they can execute once HC tasks are satisfied.

We define the demand bound function of a trigger task $\pi_i$ to be the workload of that task (running $C^h$) plus the workload of HC tasks running HI mode and having higher priority than $\pi_i$. 
\vspace{-2mm}
\[ \texttt{DEM}^c(\pi_i,{\mathcal S})= \hspace{-10mm}\sum\limits_{
\begin{array}{c} \pi_j\mid  \chi_j=HC \\
\wedge ~\Omega(\pi_j,{\mathcal S})=HI \\
\wedge ~\pi_j\in hp(\pi_i,{\mathcal S})
\end{array}} \hspace{-1mm} \Psi^H(\pi_j, [{\mathcal S},T_i]) + (C_i^h-\Lambda(\pi_i,{\mathcal S}))
\]   

To make room for the trigger task to fully execute just in case it violates its low confidence WCET, we consider $C_i^h$ instead of $C_i^l$ in $\texttt{DEM}^c()$ calculation. This can be an over-approximation but it is much safer and practical given that HC tasks non-deterministically run $C_i^h$. In case the trigger task sticks to its allotted execution time $C_i^l$, the surplus time is used to accommodate more LC tasks. 
The mode trigger task $\pi_i$ is schedulable (under the stretching pattern) if:
\[\texttt{DEM}^c(\pi_i,{\mathcal S})\le T_i-{\mathcal S}\]    

Whenever the current job of the trigger task expires \footnote{The period of the most recent ${\mathcal S}$.}, the system scheduling mode switches from \textbf{Critical} to \textbf{Normal}. The mode change instant is calculated from ${\mathcal S}$ with the time left to the period expiry of $\pi_i$, i.e. $t'={\mathcal S} + (T_i-{\mathcal S})$. 

\begin{center}
\fbox{ 
\parbox[t][2.4em][t]{0.43\textwidth}
{
$\frac{\begin{array}{c} \exists \pi_i\mid {\mathcal T}(\pi_i,{\mathcal S}) \wedge Mode({\mathcal S}+(T_i-{\mathcal S}))=\textbf{Critical} \end{array}}
{\begin{array}{c} Mode({\mathcal S}+(T_i-{\mathcal S}))\mapsto \textbf{Normal} \end{array}}$
}
}
\end{center}

Upon such a mode switch, the trigger task is refreshed for the new period where $\Omega({\mathcal T}, t')$ is set to LO and $\Lambda({\mathcal T},t')$ to 0. To such a purpose, we define the following function:

\[\mathtt{Refresh}(\pi_i,t)=(\Omega(\pi_i,t)\mapsto LO) \wedge (\Lambda(\pi_i,t)\mapsto 0) \]

\noindent where $\pi_i$ must be the most recent trigger task \footnote{${\mathcal T}(\pi_i,{\mathcal S}$) and $\forall t\in [{\mathcal S},T_i]~\forall \pi_j\neq \pi_i~\neg {\mathcal T}(\pi_j,t) $.} and $t$ is the mode switch-back instant (${\mathcal S}+(T_i-{\mathcal S})$). 


\paragraph{Stretching of \textbf{LC} task periods}
To guarantee the runtime resilience, our control-based scheduling algorithm stretches the current job periods of the LC tasks with the duration ($T_i-(t~\%~T_i)$), left to the expiry of the current release of the trigger HC task ($\pi_i$), when system mode switches to Critical (at time $t$).   
Once the system mode is switched back to Normal, one needs to absorb the stretching delay ($T_i-{\mathcal S}$) of LC tasks so that such tasks return to regular periodic dispatch.
 
\paragraph{Shrinking of \textbf{LC} task periods}
The shrinking rate of the LC task periods depends on the actual system workload and the length of the individual LC task periods. In fact, the shrinking is driven by the schedulability of the HC task running in LO mode and having the lowest priority, i.e. a priority lower than LC tasks. We consider the current job of such a HC task, and calculate first how would be the schedulability of that task according to the workload resulting from the shrinking of LC periods with a duration $\delta$. We start with $\delta$ equals to the stretching duration $(T_i-{\mathcal S})$, if the resulting workload is schedulable (using a \texttt{DEM}-based online schedulability test) then the shrinking is applied. Otherwise, we consider a tighter shrinking duration $\delta < T_i-{\mathcal S}$ and so on until the workload is schedulable. This binary process can end up having $\delta=0$ if the resulting workload is not schedulable for any potential shrinking duration.  

Let us assume a shrinking duration $\delta \le T_i-{\mathcal S}$ (the stretching duration due to the most recent trigger task). Let us assume also that $\eta$ is the instant of the system mode switch back to Normal mode. The shrinking with $\delta$ will be split over a number of periods each LC task can perform within the time left ($T_i-\eta$) to the expiry of the current job of the HC task running LO mode with lowest priority ($\pi_i$). The number of LC task ($\pi_j$) periods occurring within $[\eta,T_i]$, after shrinking with $\delta$, is given by $\frac{T_i-\eta+\delta}{T_j}$. Then the actual shrinking of each LC task ($T_j$) period is $\mu$ such that $\delta=\mu \cdot \frac{T_i-n}{T_j-\mu}$ which makes $\mu=\frac{T_j\cdot \delta}{T_i-\eta +\delta}$ \footnote{$\mu$ is the actual shrinking of each period of a given LC task $\pi_j$ whereas $\delta$ is the accumulated shrinking over [$\eta,T_i$].}. 

 We calculate first the resource demand $\mathtt{DEM}^{\delta}(\pi_i,n)$ of the HC task, running LO and having the lowest priority level, assuming the actual shrinking $\mu$ of LC task periods, from the mode change instant until the expiry of its current job period.
\[ \begin{array}{ll}
\mathtt{DEM}^{\delta}(\pi_i,\eta)= & W_{H}^h(\pi_i,\eta)+W_{H}^l(\pi_i,\eta)+W_{L}^{\delta}(\pi_i,\eta)+\\ &(C_i^l-\Lambda(\pi_i,\eta))
\end{array}
\]

\noindent The workload of LC tasks after shrinking is given as follows: 
\vspace{-2mm}
\[W_{L}^{\delta}(\pi_i,\eta)= \sum\limits_{\pi_j\mid \chi_j=LC \wedge \pi_j\in hp(\pi_i,\eta)} 
U_{L_j}^{\mu} \cdot (T_j-\mu) \cdot \lceil\frac{T_i-\eta}{T_j-\mu}\rceil
\]

Figure~\ref{fig:shrinking} depicts the period shrinking of two LC tasks for a total duration $\delta=12$. We omitted HC tasks and only the lowest priority HC task is depicted. The periods of $\pi_2$, released within interval [5,30], are shrunk with $\mu=6$ whereas the periods of $\pi_3$ are shrunk with $\mu=4$. Given that we have two periods of $\pi_2$, respectively three for $\pi_3$, within [5,30] thus the accumulated shrinking $2 \times 6=12$, respectively $3 \times 4=12$, equals $\delta$. 

\begin{figure}
\begin{center}
\caption{Example of LC task periods shrinking}
\label{fig:shrinking}
\vspace{-2mm}
\includegraphics[scale=0.33]{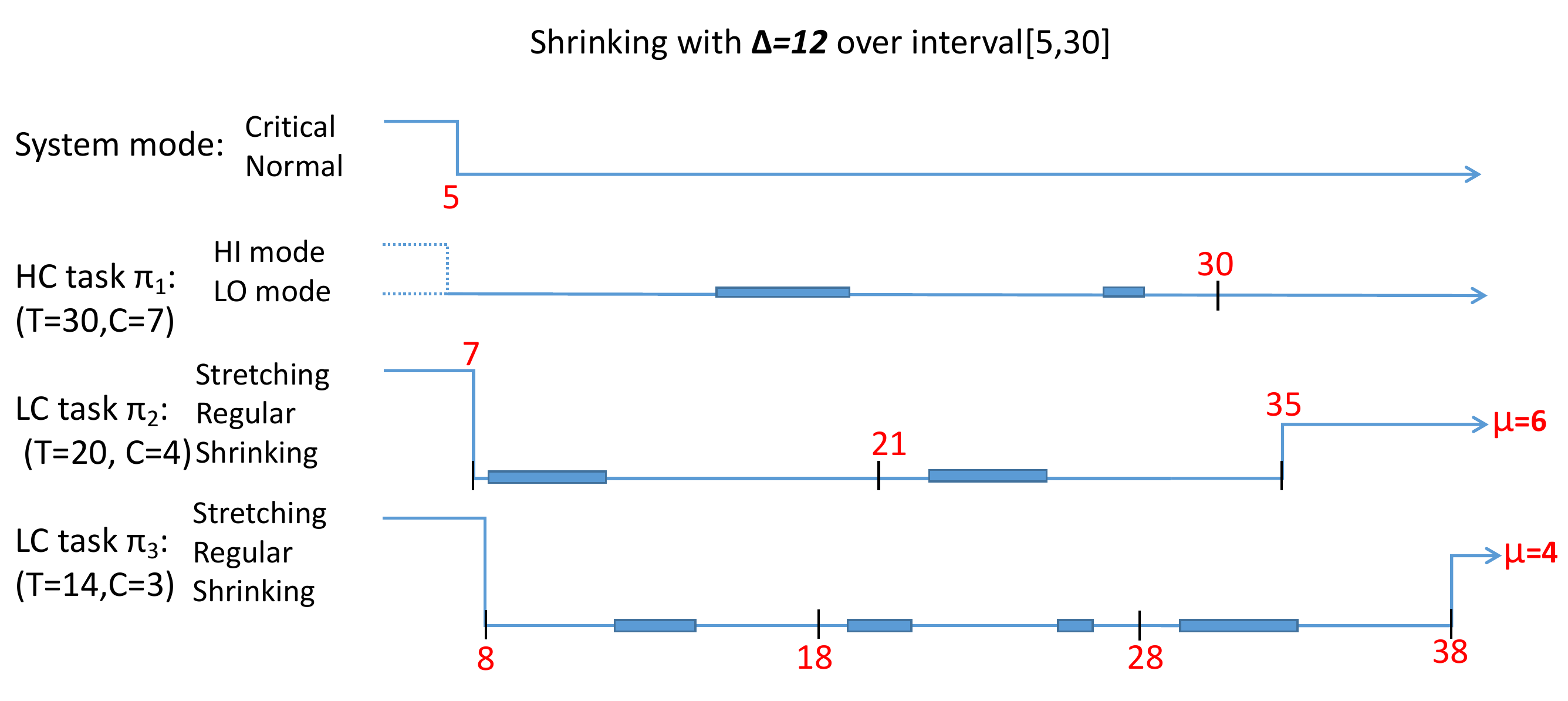}
\end{center}
\end{figure}    
 \vspace{-3mm}



%% file: algorithm.tex
Our scheduling algorithm is a control-based where the scheduling parameters and criteria (priority only, priority and criticality, priority-criticality-mode) considered to arbitrate tasks depend on the actual system workload and task modes. The overall scheduling algorithm is depicted in Algorithm~\ref{alg:algo} where $t$ is a clock variable to model the time progress. We introduce a function ${\mathit{Use()}}$ to dictate the scheduling criteria to be used during runtime, in terms of priority, default criticality and/or runtime criticality. The corresponding scheduling function ($Sched(), Sched_I()$ or $Sched_C()$) is then accordingly applied.   

Let us introduce $lp_l(t)=\pi_i\mid~\chi_i=HC~\wedge~\Omega(\pi_i,t)=LO ~\wedge \forall \pi_j~Sched_I(\pi_i,\pi_j,t)\neq \pi_i$ to be the lowest priority HC task running LO mode. Similarly, we use $lp_h(t)=\pi_i\mid~\chi_i=HC~\wedge~\Omega(\pi_i,t)=HI ~\wedge \forall \pi_j~Sched_c(\pi_i,\pi_j,t)\neq \pi_i$ to refer to the lowest priority HC task running HI mode. Whenever the execution period of a HC task expires, we refresh the task mode accordingly to be LO. 

\begin{algorithm}
$\mathit{Init()}$;

\While{True}{
  \If{$\exists \pi_i\mid Status(\pi_i,t)=Done \wedge t\%T_i=0)$}{
	\label{alg:reset}
	$Refresh(\pi_i)$;
	}
	
	\If{$\exists \pi_i\mid \chi_i=HC \wedge\Lambda(\pi_i,t)\ge C_i^l ~\wedge Status(\pi_i,t)\neq Done$}{\label{alg:LO2HI}
  $\Omega(\pi_i,t)=HI;$ \\
	$Use(Sched^I())$;{\label{alg:LO2HI+}}
	}
	
	\If{$Mode(t)=Normal \wedge \mathtt{DEM}(lp_l(t),t)<lp_l(t).T-t\%lp_l(t).T)$ }{\label{alg:N2C}
	${\mathcal T}=lp_l(t);$ \\
	${\mathcal S}=t; $\\
	$Mode(t)=Critical;$\\
	${\mathcal P}=Stretching;$\\
	$Use(Sched^c())$;\\
  \ForEach{$\pi_j \mid \chi_j=LC$}{
	  $T_j\mapsto T_j+(lp_l(t).T-t\%lp_l(t).T)$;\\
		$\delta=\delta+({\mathcal T}.T-{\mathcal S});$ {\label{alg:N2C+}}
	  }
	}
	
	\If{$Mode(t)=Critical~\wedge~\exists \pi_i \mid {\mathcal T}(\pi_i,{\mathcal S}) \wedge t\%T_i=0$}{\label{alg:C2N}
	$Mode(t)=Normal$;\\
	${\mathcal P}=Regular$;\\
	  $\eta=t$;\\
	  \If{$\exists \pi_j \mid \Omega(\pi_j,t)=HI$}{$Use(Sched_I())$;
	  }
	  \Else{$Use(Sched())$; {\label{alg:C2N+}}
		} 
	}
			
	\If{$Mode(t)=Normal~\wedge~\delta>0$}{\label{alg:Shrink}
	  	 \If{$\texttt{DEM}^{\delta}(lp_l(t),t) \le lp_l(t).T-t$}{
		    \ForEach{$\pi_j\mid \chi_j=LC$}{
			   $T_j=T_j-\mu_j$;
			  }
		  ${\mathcal P}=Shrinking$;\\
			$\delta=0$; {\label{alg:Shrink+}}
			}
	}	
 }
\caption{Elastic multimode scheduling}
\label{alg:algo} 
\end{algorithm}

The initialization function is given by:
\[ \mathit{Init()=\left \{ 
\begin{array}{ll} 
t=0 & \wedge \\
\mathit{Mode}(t)=Normal & \wedge \\
{\mathcal P}=Regular & \wedge \\
\forall i\mid \chi_i=HC~Refresh(\pi_i,t) & \wedge \\ 
\mathit{Use(Sched())} & \\
\end{array}
\right .}
\]
The statement in line \ref{alg:reset} describes when to refresh both status and mode of each HC task upon the release of a new period. The task mode switch from LO to HI is given in lines \ref{alg:LO2HI}-\ref{alg:LO2HI+}. Lines \ref{alg:N2C}-\ref{alg:N2C+} describe a system mode switch from Normal to Critical where a shrinking operation is applied. Lines \ref{alg:C2N}-\ref{alg:C2N+} describe the system mode switch back to Normal whenever the current period of the most recent trigger task expires. Lines \ref{alg:Shrink}-\ref{alg:Shrink+} outline when a shrinking operation for the LC task periods is released.

Upon each mode switch, a refreshment of some of the tasks is performed, if needed. Moreover, the scheduling function to be employed is specified using function $Use()$

In principle, a shrinking is applied as long as the stretching duration $\delta$ is not completely amortized. To simplify the algorithm, we have specified a one-go shrinking action, but the shrinking might be performed on several chunks due to preemption of the system Normal mode. This can be achieved using an extra variable to track the accumulated stretching delays.

%% file: schedulability.tex
\label{sec:analysis}
In this section we show how to analyze the schedulability of MCS running our new scheduling algorithm. Our schedulability analysis is in fact an online test checking the actual workload of the different modes and compare it against the resource supply that can be provided for each mode during a given time interval. We consider the mode switch instants to be the ground to calculate both demand and supply bound functions for our online schedulability test. This makes our schedulability test applicable no matter of how many mode switches happen during the system execution.

The ultimate goal of our algorithm and the underlying schedulability analysis is:
\begin{itemize}
\item guarantee the feasibility of HC tasks under all potential modes and patterns, i.e. $\forall t~\pi_i\mid \chi_i=HC, t ~\% ~ T_i=0\Rightarrow Status(\pi_i,t)=Done$.
\item minimize the degradation of LC tasks, and compensate for all potential degradation.
\end{itemize}

To perform the schedulability test, we define the demand bound function \texttt{DBF($\pi_i,[t,t+z]$)} to be the resource demand \texttt{DEM}($\pi_i,t$) of a HC task $\pi_i$ for the entire busy period $z$ starting at time instant $t$. We simply write:
\vspace{-2mm}
\[\texttt{DBF}(\pi_i,[t,t+z])=\texttt{DEM}(\pi_i,t|\Psi(\pi_i,[t,T_i\mapsto t+z]))\]

\noindent $\texttt{DBF}^c(\pi_i,[t,t+z])$ and $\texttt{DBF}^{\delta}(\pi_i,[t,t+z])$ are accordingly built on $\texttt{DEM}^{c}(\pi_i,t)$ and $\texttt{DEM}^{\delta}(\pi_i,t)$ respectively. $t$ is the time instant of the Normal mode release, which could be either "0" for the initial system release or a time instant where the system mode switches back to Normal.  

A given system remains under Normal mode as long as all HC tasks are schedulable, $\mathtt{DBF}()$ of the lowest priority HC task $\pi_i$ does not exceed the potential resource supply for the time interval $[t,T_i]$. 
To check schedulability, regardless of the individual task modes, we analyze $\mathtt{DBF}()$ of the lowest priority HC task.   

\begin{theorem}[Schedulability under Normal mode] \label{theo1}
The HC taskset is schedulable when the system runs in mode \textbf{Normal}, with at least one \textbf{HC} task under mode \textbf{LO}, if the following holds: 
\[\begin{array}{c} 
\forall t~Mode(t)=Normal~\forall \pi_i\mid \Omega(\pi_i,t)=LO ~\wedge~lp_l(t)=\pi_i \\ ~~~~~~~~~~~~~~~~~~ \text{and}~  \texttt{DBF}(\pi_i,[t,t+z]) \le z
\end{array}\]
\end{theorem}
\begin{proof}
It is trivial. Given that $\pi_i$ is the least priority ($lp_l(t)$) HC task ($\Omega(\pi_i,t)=LO$), then $\forall \pi_j\neq \pi_i~\pi_j\in hp(\pi_i,t)$. Since we only consider fixed priority policies, thus $lp_l(t)=\pi_i \Rightarrow lp_l(t'\in [t,t+z])=\pi_i$, i.e $\pi_i$ remains the lowest priority HC task over [t,t+z]. From $\texttt{DBF}(\pi_i,[t,t+z])$ definition $W^h_H(\pi_i,t)$ and $W^l_H(\pi_i,t)$  \footnote{With $\Psi^H(\pi_i,[t,z])$ and $\Psi^L(\pi_i,[t,z])$ calculated for the entire busy period.} include the workload of each newly released HC job in the time interval [t,t+z] having either a higher priority ($\pi_j\in hp(\pi_i,t) \wedge \Omega(\pi_j,t)=LO$) or a higher task mode ($\Omega(\pi_j,t)=HI$), and the execution budget left for the actual period of time instant $t$ ($C^L_i-\Lambda(\pi_i,t)$). Thus, if $\pi_i$ is schedulable then $\forall \pi_j\mid Sched({i,j},t'\in[t,t+z])=\pi_j \wedge \Omega(\pi_j)\ge \Omega(\pi_i)$ is schedulable. 
\end{proof}
This Theorem implies that, in case the lowest priority task is a high critical, the schedulability test includes all HC and LC tasks. Thus, the schedulability of HC tasks implies the schedulability of the entire task set. 

In case the system is in Normal mode but all HC tasks run mode HI, there is no point to consider LC tasks as any HC task has priority over all LC tasks. 

\begin{theorem}[Schedulability when all \textbf{HC} tasks run \textbf{HI} mode] 
The \textbf{HC} taskset is schedulable when the system runs in mode \textbf{Normal}, with all \textbf{HC} tasks under mode \textbf{HI}, if the following holds: 
\[\begin{array}{c}
\forall t~Mode(t)=Normal~\wedge~\forall \pi_j~\Omega(\pi_j,t)=HI \\ ~~~~~~~~~~~~~~~~~~ \text{and}~
  \texttt{DBF}^c(lp_h(t),[t,t+z]) \le z
\end{array}\]
\end{theorem}
\begin{proof} It is trivial.
\end{proof}

In a similar way, the schedulability of the HC taskset under shrinking pattern is defined by the schedulability of the lowest priority HC task running LO mode. This is because such a task is comparable to LC tasks, thus it can be affected by the shrinking workload.  

\begin{theorem}[Schedulability under Shrinking pattern]
\textbf{HC} taskset is schedulable when the system runs a shrinking with a delay $\delta$ if: 
\[\begin{array}{c}
\forall t~Mode(t)=Normal ~\wedge~{\mathcal P}=Shrinking \Rightarrow \\ ~~~~~~~~~~~~~~~~~~ 
  ~~~~ \texttt{DBF}^{\delta}(lp_l(t),[t,t+z]) \le z
\end{array}\]
\end{theorem}
\begin{proof} It is similar to that of Theorem.~\ref{theo1}.
\end{proof}
Again, this theorem implies not only the schedulability of HC tasks but the schedulability of the entire task set in case the lowest priority task of $\Pi$ is a HC task.
 
Whenever a HC task, running in mode LO, is jeopardized to miss its deadline under mode Normal our scheduling algorithm anticipates a system mode change to \textbf{Critical}. Thus, HC taskset is schedulable under Critical mode if the lowest priority HC task running in mode LO, known as a trigger task, is schedulable. 

\begin{theorem}[Schedulability under critical mode]
\textbf{HC} taskset is schedulable when the system runs \textbf{Critical} mode if: 
\[\begin{array}{c}
\forall t~Mode(t)=Critical, ~\exists \pi_i\mid \Omega(\pi_i,t)=LO ~\wedge~ \\
~~~~ \forall \pi_j\mid \chi_j=HC,~ Sched_C(\pi_i,\pi_j,t)\neq \pi_i \Rightarrow \\ 
\texttt{DBF}^c(\pi_i,[t,t+z]) \le z
\end{array}\]
\end{theorem} 
\begin{proof} The condition $\forall \pi_j\mid \chi_j=HC~ Sched_C(\pi_i,\pi_j,t)\neq \pi_i$ implies that $\pi_i$ is either the lowest  priority HC task or the HC task having the lowest task mode ($\Omega(\pi_i,t)=LO$) given that $Sched_C()$ relies on both task runtime mode and priority. By definition $\texttt{DBF}^c(\pi_i,[t,t+z])$, includes the workload of all HC tasks $\pi_j\mid  \chi_j=HC \wedge ~\Omega(\pi_j,t)=HI ~\wedge ~\pi_j\in hp(\pi_i,t)$. Thus, if $\pi_i$ is schedulable then any other HC task will be schedulable.  
\end{proof}

%% file: casestudy.tex
\label{sec:casestudy}
To study the applicability and performance of our multimode scheduling algorithm and show the underlying schedulability analysis, we have analyzed an actual example from the avionic domain \cite{Dodd2006}. The most relevant attributes of the task set description are given in Table~\ref{tab:casestudy}.

\begin{table}
\centering
\caption{Task attributes of the case study}
\label{tab:casestudy}
\begin{tabular}{|l|c|c|c|c|c|}
\hline 
\textbf{Task} & $\chi$ & $T$ & $C^l$ & $C^h$ & $\rho$ \\ \hline
Aircraft flight data($\pi_1$) & HC & 55 & 8 & 8.9  & 6\\ \hline
Steering($\pi_2$) & HC& 80 & 6 & 6.3 & 9\\ \hline 
Target tracking($\pi_3$) & HC & 40 & 4  & 4.2 & 3 \\ \hline 
Target sweetening($\pi_4$) & HC & 40 & 2 & 2 & 4 \\ \hline  
AUTO/CCIP toggle($\pi_5$) & HC & 200 & 1 & 1 & 12\\ \hline 
Weapon trajectory($\pi_6$) & HC & 100 & 7 & 7.5 & 10 \\ \hline 
Reinitiate trajectory($\pi_7$) & LC & 400 & 6.5 & - & 14 \\ \hline
Weapon release($\pi_8$) & HC & 10 & 1 & 1.2 & 1\\ \hline
HUD display($\pi_9$) & LC & 52 & 6 & - & 7 \\ \hline 
MPD tactical display($\pi_{10}$) & LC & 52 & 8 & - & 8\\ \hline 
Radar tracking($\pi_{11}$) & HC & 40 & 2 & 2.2 & 2 \\ \hline 
HOTAS bomb button ($\pi_{12}$) & LC & 40 & 1 & - & 5 \\ \hline
Threat response display($\pi_{13}$) & LC & 100 & 3 & - & 11\\ \hline
Poll RWR($\pi_{14}$) & LC & 200 & 2 & - & 13 \\ \hline 
Perodic BIT($\pi_{15}$) & LC & 1000 & 5 & - & 15\\ \hline 
\end{tabular}
\end{table}

 We have synthetically calculated $C^h$ from $C^l$ by considering the worst case response time of data fetching. The original taskset description of \cite{Dodd2006} states how many data each task exchanges during each period. The best case response time of data fetching is instantaneous whereas the worst case response time is $20\mu s$ for data words, $40\mu s$ for a command and $40\mu s$ for a status. The scheduling policy adopted to schedule the task set is FP (fixed priority). 

To analyze the case study, we have mechanized the system model and scheduling algorithms in Uppaal \cite{Behrmann2004}. When we run the taskset using a classic priority-based scheduling, tasks $\pi_{10}$ and $\pi_{11}$ miss their deadlines making thus the system not schedulable. When the system runs fixed priority policy with \textit{task level} scheduling mode only, task $\pi_{10}$ \textit{misses its deadline} (response time~106). 

When the taskset runs the \textit{system-level} scheduling mode, all HC tasks meet their deadlines whereas multiple LC jobs are discarded to achieve the schedulability of HC tasks. The number of LC task jobs discarded is depicted in Fig.~\ref{fig:comparison}.

When the system runs our \textit{multimode scheduling algorithm}, all the high criticality tasks meet their deadlines. To achieve the schedulability of the HC tasks, our scheduling algorithm postpones the execution of some of the LC tasks. We consider each postponing operation with a delay longer than the corresponding LC task slack time to be a discard case. This is because a delay longer than the available slack time will absolutely lead the task execution to miss its deadline. The number of LC task jobs discarded by our algorithm  is depicted in Figure~\ref{fig:comparison}. 
 
Compared to the state of the art, for the given case study, our multimode scheduling algorithm guarantees the schedulability of all HC tasks whereas Task-level scheduling algorithms do not. Moreover, the discard rate of the LC task jobs achieved by our algorithm is 1.0\% to 4.58\% whereas the discard rate achieved by the state of the art system-level bi-mode scheduling \cite{zeroslack,Facs16} is 2.1\% to 11.5\%. The discard rate is calculated to be the number of jobs discarded to the total number of jobs released. 

An important observation from this experiment is that, although the proposed algorithm achieves less discards to low criticality tasks, it requires around 30\% extra overhead compared to most of the state of the art algorithms. By overhead we mean the data size to track the system runtime and the time to process such data. Thus, the combination of task-level and system-level mode switches is \textit{not efficient} in  making real-time scheduling decisions. Another observation is that the compensation of LC tasks is slow given that LC tasks have the period lengths comparable to the period of the lowest priority HC task.     

\begin{figure}
\centering
\caption{Comparison of the LC task jobs discarded}
\label{fig:comparison}
\includegraphics[scale=0.42]{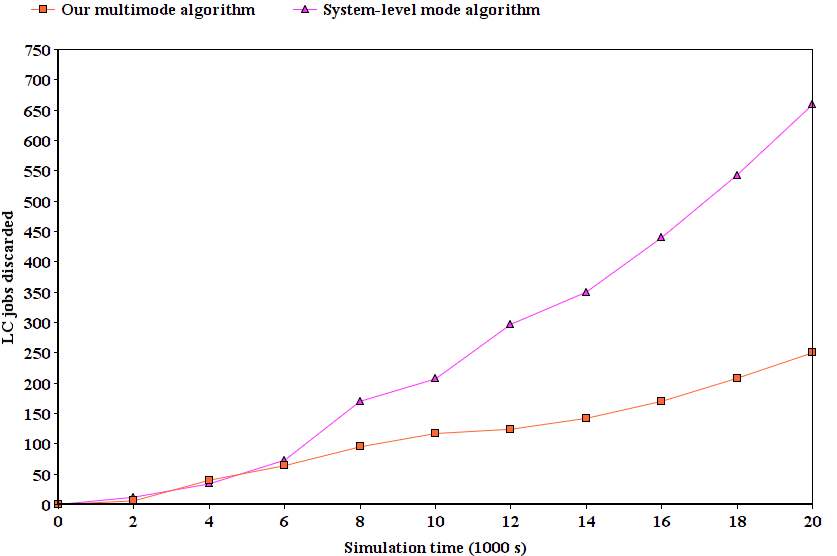}
\end{figure}

%% file: conclusion.tex
\label{sec:conclusion}
This paper introduced a flexible multimode scheduling algorithm for mixed criticality systems by combining the system-level and task-level mode switch techniques. The proposed algorithm relies on a job-level mode switch, where we restrict the HC task behavior to only the job that either exceeds its low confidence WCET or triggers a system mode switch. This technique provides an exact schedulability test for the system mode switches. Low criticality tasks are not discarded under critical mode, rather their periods are stretched to loosen the underlying workload. Such tasks are later compensated for the degradation, due to stretching, by shrinking their subsequent periods accordingly. We have mechanized our new multimode scheduling algorithm in Uppaal and analyzed an actual avionic system component as a case study.

The efficiency of our elastic algorithm remains in the fact that considering a short range load calculation of high criticality tasks leads to accurate and non-aggressive system mode switches.  

Although combining task-level and system-level scheduling modes offers a higher flexibility and accuracy, it experiences a heavy overhead to calculate real-time scheduling decisions. Thus, such a combination is not suitable for the scheduling of safety critical real-time systems.

As a future work, we aim to study potential optimizations of the proposed algorithm overhead. 

%% file: main.bbl
\begin{thebibliography}{10}

\bibitem{iso}
{ISO 26262-1:2011 Road vehicles--Functional safety}.
\newblock Technical report, {ISO}, 2011.

\bibitem{Ouedraogo18}
K.~Abel~Ouedraogo, J.~Beugin, E.~M. El~Koursi, J.~Clarhaut, D.~Renaux, and
  F.~Lisiecki.
\newblock Toward an application guide for safety integrity level allocation in
  railway systems.
\newblock {\em International Journal of Risk Analysis}, 2018.

\bibitem{Burns13}
A.Burns and S.~Baruah.
\newblock Towards a more practical model for mixed criticality systems.
\newblock In {\em Workshop on Mixed- Criticality Systems (co-located with
  RTSS)}, 2013.

\bibitem{Baruah12}
S.~Baruah, V.~Bonifaci, G.~DAngelo, H.~Li, A.~Marchetti-Spaccamela, S.~van~der
  Ster, and L.~Stougie.
\newblock The preemptive uniprocessor scheduling of mixed-criticality
  implicit-deadline sporadic task systems.
\newblock In {\em ECRTS 2012}, pages 145--154, 2012.

\bibitem{Baruah16}
S.~Baruah, A.~Burns, and Z.~Guo.
\newblock Scheduling mixed-criticality systems to guarantee some service under
  all non-erroneous behaviors.
\newblock In {\em ECRTS 2016}, pages 131--138, July 2016.

\bibitem{Baruah90}
S.~K. Baruah, L.~E. Rosier, and R.~R. Howell.
\newblock Algorithms and complexity concerning the preemptive scheduling of
  periodic, real-time tasks on one processor.
\newblock {\em Real-Time Syst.}, 2(4):301--324, Oct. 1990.

\bibitem{Bate15}
I.~Bate, A.~Burns, and R.~I. Davis.
\newblock A bailout protocol for mixed criticality systems.
\newblock In {\em ECRTS 2015}, pages 259--268, July 2015.

\bibitem{Behrmann2004}
G.~Behrmann, A.~David, and K.~G. Larsen.
\newblock {\em A Tutorial on Uppaal}, pages 200--236.
\newblock Springer Berlin Heidelberg, Berlin, Heidelberg, 2004.

\bibitem{Boudjadar17}
J.~{Boudjadar}.
\newblock An efficient energy-driven scheduling of dvfs-multicore systems with
  a hierarchy of shared memories.
\newblock In {\em IEEE/ACM 21st DS-RT Conference}, pages 1--8, 2017.

\bibitem{Burns17}
A.~Burns and R.~I. Davis.
\newblock A survey of research into mixed criticality systems.
\newblock {\em ACM Comput. Surv.}, 50(6):82:1--82:37, Nov. 2017.

\bibitem{Burns18}
A.~Burns, R.~I. Davis, S.~Baruah, and I.~Bate.
\newblock Robust mixed-criticality systems.
\newblock {\em {IEEE} Transactions on Computers}, To appear, 2018.

\bibitem{Burns10}
A.~Burns and B.~Littlewood.
\newblock Reasoning about the reliability of multi-version, diverse real-time
  systems.
\newblock In {\em 2010 31st IEEE Real-Time Systems Symposium}, pages 73--81,
  2010.

\bibitem{zeroslack}
D.~de~Niz, K.~Lakshmanan, and R.~Rajkumar.
\newblock On the scheduling of mixed-criticality real-time task sets.
\newblock In {\em RTSS'09}, pages 291--300, 2009.

\bibitem{Dodd2006}
R.~Dodd.
\newblock Coloured petri net modelling of a generic avionics missions computer.
\newblock Technical report, Department of Defence, Australia, Air Operations
  Division, 2006.

\bibitem{Arvind13}
A.~Easwaran.
\newblock Demand-based scheduling of mixed-criticality sporadic tasks on one
  processor.
\newblock In {\em 2013 IEEE 34th Real-Time Systems Symposium}, pages 78--87,
  Dec 2013.

\bibitem{Ekberg14}
P.~Ekberg and W.~Yi.
\newblock Bounding and shaping the demand of generalized mixed-criticality
  sporadic task systems.
\newblock {\em Real-Time Systems}, 50(1):48--86, Jan 2014.

\bibitem{Erickson15}
J.~P. Erickson, N.~Kim, and J.~H. Anderson.
\newblock Recovering from overload in multicore mixed-criticality systems.
\newblock In {\em 2015 IEEE International Parallel and Distributed Processing
  Symposium}, pages 775--785, May 2015.

\bibitem{Fleming14}
T.~Fleming and A.~Burns.
\newblock Incorporating the notion of importance into mixed criticality
  systems.
\newblock In {\em Proceedings of Workshop on Mixed Criticality Systems},
  page~33, 2014.

\bibitem{Gettings15}
O.~Gettings, S.~Quinton, and R.~I. Davis.
\newblock Mixed criticality systems with weakly-hard constraints.
\newblock In {\em Proceedings of the 23rd International Conference on Real Time
  and Networks Systems}, RTNS '15, pages 237--246. ACM, 2015.

\bibitem{Xiaozhe16}
X.~Gu and A.~Easwaran.
\newblock Dynamic budget management with service guarantees for
  mixed-criticality systems.
\newblock In {\em Proceedings of the IEEE Real-Time Systems Symposium}, pages
  47--56. IEEE, 2016.

\bibitem{Xiaozhe15}
X.~Gu, A.~Easwaran, K.-M. Phan, and I.~Shin.
\newblock Resource efficient isolation mechanisms in mixed-criticality
  scheduling.
\newblock In {\em Proceedings of the Euromicro Conference on Real-Time
  Systems}, pages 13--24, July 2015.

\bibitem{Howard17}
G.~Howard, M.~Butler, J.~Colley, and V.~Sassone.
\newblock Formal analysis of safety and security requirements of critical
  systems supported by an extended stpa methodology.
\newblock In {\em 2017 IEEE European Symposium on Security and Privacy
  Workshops (EuroS PW)}, pages 174--180, 2017.

\bibitem{Hu16}
B.~Hu, K.~Huang, P.~Huang, L.~Thiele, and A.~Knoll.
\newblock On-the-fly fast overrun budgeting for mixed-criticality systems.
\newblock In {\em Proceedings of the IEEE \& ACM International Conference on
  Embedded Software}, pages 1--10. IEEE, 2016.

\bibitem{Huang14}
P.~Huang, G.~Giannopoulou, N.~Stoimenov, and L.~Thiele.
\newblock Service adaptions for mixed-criticality systems.
\newblock In {\em In Proceedings of ASP-DAC}, 2014.

\bibitem{Huang15}
P.~Huang, P.~Kumar, G.~Giannopoulou, and L.~Thiele.
\newblock Run and be safe: Mixed-criticality scheduling with temporary
  processor speedup.
\newblock In {\em DATE 2015}, 2015.

\bibitem{Huang13}
P.~Huang, P.~Kumar, N.~Stoimenov, and L.~Thiele.
\newblock Interference constraint graph - a new specification for
  mixed-criticality systems.
\newblock In {\em ETFA 2013}, pages 1--8, 2013.

\bibitem{Huyck12}
P.~Huyck.
\newblock Arinc 653 and multi-core microprocessors; considerations and
  potential impacts.
\newblock In {\em DASC'12}, pages 6B4--1--6B4--7, 2012.

\bibitem{Jan13}
M.~Jan, L.~Zaourar, and M.~Pitel.
\newblock Maximizing the execution rate of low criticality tasks in mixed
  criticality system.
\newblock In {\em Proceedings of Workshop on Mixed-Criticality, RTSS 2013},
  pages 43--48, 2013.

\bibitem{Lee17}
J.~Lee, H.~S. Chwa, L.~T.~X. Phan, I.~Shin, and I.~Lee.
\newblock {MC-ADAPT}: Adaptive task dropping in mixed-criticality scheduling.
\newblock {\em ACM Trans. Embed. Comput. Syst.}, 16(5s):163:1--163:21, Sept.
  2017.

\bibitem{Liu18}
D.~Liu, N.~Guan, J.~Spasic, G.~Chen, S.~Liu, T.~Stefanov, and W.~Yi.
\newblock Scheduling analysis of imprecise mixed-criticality real-time tasks.
\newblock {\em IEEE Transactions on Computers}, 2018.

\bibitem{Liu16}
D.~Liu, J.~Spasic, N.~Guan, G.~Chen, S.~Liu, T.~Stefanov, and W.~Yi.
\newblock Edf-vd scheduling of mixed-criticality systems with degraded quality
  guarantees.
\newblock In {\em 2016 IEEE Real-Time Systems Symposium (RTSS)}, pages 35--46,
  2016.

\bibitem{Loefwenmark16}
A.~L{\"{o}}fwenmark and S.~Nadjm-Tehrani.
\newblock Understanding shared memory bank access interference in multi-core
  avionics.
\newblock In {\em Proceedings of WCET'16}, OpenAccess Series in Informatics
  (OASIcs), 2016.

\bibitem{Facs16}
B.~Madzar, J.~Boudjadar, J.~Dingel, T.~E. Fuhrman, and S.~Ramesh.
\newblock Formal analysis of predictable data flow in fault-tolerant multicore
  systems.
\newblock In {\em {FACS} '16}, pages 153--171, 2016.

\bibitem{Papadopoulos18}
A.~V. Papadopoulos, E.~Bini, S.~Baruah, and A.~Burns.
\newblock {AdaptMC}: {A} control-theoretic approach for achieving resilience in
  mixed-criticality systems.
\newblock In {\em 30th Euromicro Conference on Real-Time Systems, {ECRTS}
  2018}, pages 14:1--14:22, 2018.

\bibitem{Park2011}
T.~Park and S.~Kim.
\newblock Dynamic scheduling algorithm and its schedulability analysis for
  certifiable dual-criticality systems.
\newblock In {\em Proceedings of EMSOFT '11}, pages 253--262. ACM, 2011.

\bibitem{Pathan17}
R.~M. Pathan.
\newblock {Improving the Quality-of-Service for Scheduling Mixed-Criticality
  Systems on Multiprocessors}.
\newblock In {\em 29th Euromicro Conference on Real-Time Systems (ECRTS 2017)},
  volume~76 of {\em Leibniz International Proceedings in Informatics (LIPIcs)},
  pages 19:1--19:22, Dagstuhl, Germany, 2017.

\bibitem{Ren15}
J.~Ren and L.~T.~X. Phan.
\newblock Mixed-criticality scheduling on multiprocessors using task grouping.
\newblock In {\em 27th Euromicro Conference on Real-Time Systems (ECRTS)},
  pages 25--34, July 2015.

\bibitem{Santy12}
F.~Santy, L.~George, P.~Thierry, and J.~Goossens.
\newblock Relaxing mixed-criticality scheduling strictness for task sets
  scheduled with {FP}.
\newblock In {\em ECRTS '12}, pages 155--165, July 2012.

\bibitem{Santy13}
F.~Santy, G.~Raravi, G.~Nelissen, V.~Nelis, P.~Kumar, J.~Goossens, and
  E.~Tovar.
\newblock Two protocols to reduce the criticality level of multiprocessor
  mixed-criticality systems.
\newblock In {\em Proceedings of RTNS '13}, pages 183--192, New York, NY, USA,
  2013. ACM.

\bibitem{Su16}
H.~Su, P.~Deng, D.~Zhu, and Q.~Zhu.
\newblock Fixed-priority dual-rate mixed-criticality systems: Schedulability
  analysis and performance optimization.
\newblock In {\em Proceedings of RTCSA}, 2016.

\bibitem{Su14}
H.~Su, N.~Guan, and D.~Zhu.
\newblock Service guarantee exploration for mixed-criticality systems.
\newblock In {\em Proceedings of RTCSA}, pages 1--10, Aug 2014.

\bibitem{Su13}
H.~Su and D.~Zhu.
\newblock An elastic mixed-criticality task model and its scheduling algorithm.
\newblock In {\em 2013 Design, Automation Test in Europe Conference Exhibition
  (DATE)}, pages 147--152, 2013.

\bibitem{Su16_2}
H.~Su, D.~Zhu, and S.~Brandt.
\newblock An elastic mixed-criticality task model and early-release edf
  scheduling algorithms.
\newblock {\em ACM Trans. Des. Autom. Electron. Syst.}, 22(2):28:1--28:25, Dec.
  2016.

\bibitem{Su2015}
H.~Su, D.~Zhu, and J.~Zhu.
\newblock On the implementation of rt-fair scheduling framework in linux.
\newblock In {\em IUCC 2015}, pages 1258--1265, 2015.

\bibitem{Vestal07}
S.~Vestal.
\newblock Preemptive scheduling of multi-criticality systems with varying
  degrees of execution time assurance.
\newblock In {\em 28th IEEE International Real-Time Systems Symposium (RTSS
  2007)}, pages 239--243, 2007.

\end{thebibliography}
